\documentclass[MSNbibl,number,citesort,seceqn,dvips]{arxbj}

\aid{0}
\volume{20}
\issue{1}
\pubyear{2014}
\firstpage{245}
\lastpage{264}
\doi{10.3150/12-BEJ484} 

\makeatletter
\newtheorem{theo}{Theorem}[section]
\newtheorem{cor}{Corollary}[section]
\newproclaim{defi}{Definition}[section]
\newremark{exam}{Example}[section]
\newtheorem{lem}{Lemma}[section]
\newtheorem{prop}{Proposition}[section]
\newremark{rem}{Remark}[section]

\renewcommand{\Pr}{\mathbb{P}}
\newcommand{\E}{\mathbb{E}}
\newcommand{\Var}{\operatorname{Var}}
\newcommand{\sign}{\operatorname{sign}}
\newcommand{\R}{\mathbb{R}}
\newcommand{\essinf}{\operatorname{ess\inf}}
\newcommand{\esssup}{\operatorname{ess\sup}}
\newcommand{\CH}{\mathcal{H}}

\newcommand{\ud}{\mathrm{d}}
\newcommand{\IP}{\operatorname{IP}(\mu;\delta,\beta,\gamma)}
\newcommand{\IPq}{\operatorname{IP}(\mu;q)}
\newcommand{\IPtilde}{\operatorname{IP}(\widetilde{\mu};\widetilde{q})}
\newcommand{\IPqk}{\operatorname{IP}(\mu_k;q_k)}
\newcommand{\lead}{\operatorname{lead}}

\makeatother

\begin{document}
\begin{frontmatter}

\title{Strengthened Chernoff-type variance bounds}
\runtitle{Strengthened Chernoff-type variance bounds}

\begin{aug}
\author{\fnms{G.} \snm{Afendras}\thanksref{e1}\ead[label=e1,mark]{g\_afendras@math.uoa.gr}} \and
\author{\fnms{N.} \snm{Papadatos}\corref{}\thanksref{e2}\ead[label=e2,mark]{npapadat@math.uoa.gr}\ead[label=u2,url]{users.uoa.gr/\textasciitilde npapadat/}}
\runauthor{G. Afendras and N. Papadatos} 
\address{Department of Mathematics, Section of Statistics and O.R.,
University of Athens, Panepistemiopolis, 157 84 Athens, Greece.
\printead{e1,e2}; \printead{u2}}
\end{aug}

\received{\smonth{5} \syear{2012}}
\revised{\smonth{10} \syear{2012}}

%
\begin{abstract}
Let $X$ be an absolutely continuous random
variable from the integrated Pearson family and assume that $X$ has
finite moments of any order.
Using some properties of the associated orthonormal polynomial system,
we provide a class of strengthened Chernoff-type variance bounds.
\end{abstract}

%

\end{frontmatter}

\section{Introduction}\label{sec.1}

Let $Z$ be a standard normal random variable and $g\dvtx \R\to\R$
any absolutely continuous function with derivative $g'$ such that
$\E(g'(X))^2<\infty$.
Chernoff \cite{Cher},
using Hermite polynomials, proved that
%
\begin{equation}\label{Chernoff}
\Var g(Z)\leq\E\bigl(g'(Z)\bigr)^2;
\end{equation}
see, also, Nash~\cite{Nash} and
Brascamp and Lieb~\cite{BL}.
In (\ref{Chernoff}), the
equality holds if and only if $g$
is a polynomial of degree at most one --
a linear function. This inequality plays an important role in
the isoperimetric problem, as well as to several areas
in probability
and statistics. It has been extended and
generalized by many authors, including
\cite{Chen,Cac,BorUtev,Klaassen,CP1,Pap1,John,HK,PP,OlkShepp,P-Rao,WZ,AP}.
On the other hand, Cacoullos \cite{Cac}
showed
the inequality
%
\begin{equation}
\label{Cacoullos} \Var g(Z)\geq\E^2 g'(Z), %
\end{equation}
in which the equality again holds if and only if $g$ is linear.

In this article, we provide improvements
on Chernoff's bound. In particular,
an application of
the main result (Theorem~\ref{theo.main}, $n=1$)
to $Z$ yields the inequality
%
\begin{equation}\label{StrongNormal}
\Var g(Z)\leq\tfrac12 \E^2 g'(Z) +
\tfrac12 \E\bigl(g'(Z)\bigr)^2, %
\end{equation}
in which the equality holds if and only if $g$ is a polynomial of degree
at most two. In view of (\ref{Cacoullos}), it is clear that the
upper bound in (\ref{StrongNormal}) improves the one given in
(\ref{Chernoff}) and, in fact, it is strictly better,
unless $g$ is linear. The difference in
right-hand sides (\ref{Chernoff}) minus (\ref{StrongNormal}) is equal
to $\frac12\Var g'(Z)$, indicating the magnitude of this improvement.

Similar bounds are valid for all distributions that will
be studied in the sequel, namely, the Beta, Gamma and Normal.
The main result applies to any Pearson (more precisely,
integrated Pearson) random variable possessing
moments of any order. Hence, Theorem~\ref{theo.main}
also improves the
bounds for Beta random variables, given by
\cite{P-Rao,WZ}.
The integrated Pearson distributions
are defined as follows,
see
\cite{John,APP2,AP,AP2}:
%
\begin{defi}[(Integrated Pearson family)]\label{def.IP}
Let $X$ be an absolutely continuous random variable with density $f$
and finite mean $\mu=\E X$. We say that
$X$ (or its density $f$) belongs to the integrated Pearson
family
if there exists a quadratic polynomial
$q(x)=\delta x^2+\beta x+\gamma$
with $\delta,\beta,\gamma\in\R$,
$|\delta|+|\beta|+|\gamma|>0$,
such that
\begin{equation}\label{qua1}
\int_{-\infty}^{x}(\mu-t)f(t)
\,\ud{t}=q(x)f(x) \qquad \mbox{for all } x\in\R.
\end{equation}
This fact will be denoted by
\begin{equation}\label{IP}
X\sim\IPq\mbox{ or } f\sim\IPq \quad \mbox{or, more explicitly},\quad  X\mbox{ or } f\sim\IP.
\end{equation}
\end{defi}

In the sequel, whenever we claim that $X$ or $f\sim\IP$, it will
be understood
that the density $f$ has been chosen in $C^{\infty}(\alpha,\omega)$ and
is vanishing outside
$(\alpha,\omega)$, where $(\alpha,\omega):=(\essinf(X),\esssup(X))$
is the interval support of $X$; see \cite{AP2}, Proposition~2.1.
Consider an arbitrary real polynomial $q$ with $\deg(q)\leq2$
such that the set $S^{+}(q):=\{x\dvtx q(x)>0\}$ is nonempty.
It can be shown that for any $\mu\in S^{+}(q)$ (i.e.,
with $q(\mu)>0$), there exists a unique (up to equality
in distribution) random variable $X$ with mean $\mu$
such that its density $f$ satisfies (\ref{qua1});
see \cite{AP2}, Section~2.

Many commonly used
continuous distributions are members of the
integrated Pearson family,
for example, Normal, Beta, Gamma and Negative Gamma. This list also includes
Pareto (with density $f(x)=a(x+1)^{-a-1}$, $x>0$, and parameter $a>1$),
Reciprocal Gamma (with density $f(x)=\lambda^a x^{-a-1}\mathrm{e}^{-\lambda
/x}/\Gamma(a)$, $x>0$, and parameters $a>1$ and $\lambda>0$),
$F_{n,m}$ (with $m>2$) and $t_n$ (with $n>1$)
distributions, their location-scale families
and their negatives -- see Table~2.1 in \cite{AP2} for a complete description.
The proof of the main result is based on specific
properties of the associated orthogonal polynomials
that can be found in \cite{AP2}.
For easy reference,
all required results are reviewed in \hyperref[app]{Appendix}.

\section{Preliminaries}\label{sec2}
The following definition will be used in the sequel.
%
%
\begin{defi}[(Cf. \cite{AP}, page 3629)]\label{def.class}
Assume that
$X\sim\IPq$
and denote by $q(x)=\delta x^2+\beta x+\gamma$ its
quadratic polynomial.
Let $(\alpha,\omega)$ be the support of $X$ and
fix an integer $n\in\{1,2,\ldots\}$.
We shall denote by
$\CH^n(X)$ the
class of functions
$g\dvtx (\alpha,\omega)\to\R$ satisfying the
following two properties:
\begin{enumerate}[H$_2$:]
%
\item[H$_1$:]
For each $k\in\{0,1,\ldots,n-1\}$,
$g^{(k)}$ (with $g^{(0)}=g$)
is an absolutely continuous function with
a.s. derivative $g^{(k+1)}$. That is,
$g\in C^{n-1}(\alpha,\omega)$ and the function
$g^{(n-1)}\dvtx (\alpha,\omega)\to\R$, with
\[
g^{(n-1)}(x):=\frac{\ud^{n-1}g(x)}{\ud{x}^{n-1}}, \qquad \alpha<x<\omega,
\]
is absolutely continuous in $(\alpha,\omega)$
with a.s. derivative $g^{(n)}$ such that
\[
g^{(n-1)}(y)-g^{(n-1)}(x) =\int_{x}^{y}
g^{(n)}(t)\,\ud{t} \qquad \mbox{for every compact interval } [x,y]\subseteq(
\alpha,\omega).
\]
\item[H$_2$:]
$\E q^n(X) (g^{(n)}(X))^2<\infty$.
\end{enumerate}
%
%
%
\end{defi}
%

Also, we denote by $\CH^{0}(X)$ and $\CH^{\infty}(X)$ the following
classes of functions:
\begin{eqnarray*}
\CH^{0}(X)&:=&L^2(\R,X)\equiv \bigl\{g\dvtx (\alpha,\omega)\to
\R, \mbox {Borel measurable, such that}\, \Var g(X)<\infty \bigr\};
\\
\CH^{\infty}(X)&:=&\bigcap_{n=0}^{\infty}\CH^{n}(X)
= \bigl\{g\in C^{\infty}(\alpha,\omega)\dvtx  \E q^n(X)
\bigl(g^{(n)}(X)\bigr)^2<\infty \mbox{ for all } n=0,1,\ldots
\bigr\}.
\end{eqnarray*}
It is clear that
$\E^2q^{n}(X)|g^{(n)}(X)|\leq\E q^{n}(X)
\E q^{n}(X)(g^{(n)}(X))^2 <\infty$,
provided $\E|X|^{2n}<\infty$
(equivalently,
$\delta<1/(2n-1)$; see Lemma~\ref{lem.moments}).
On the other hand, under suitable
moment conditions on $X$, the assumption H$_2$
implies that
$\E q^i(X)(g^{(i)}(X))^2<\infty$ for
all $i\in\{0,1,\ldots,n\}$.
In particular, if all moments exist (equivalently,
if $\delta\leq0$),
then
\[
L^2(\R,X)=\CH^0(X)\supseteq\CH^{1}(X)
\supseteq \CH^{2}(X)\supseteq\cdots\supseteq\CH^{\infty}(X),
\]
that is,
$\CH^{n}(X)=\bigcap_{i=0}^n \CH^i(X)$ for all $n$.
In order to verify this fact
we first show a lemma.

\begin{lem}\label{lem.deriv.l2}
If $X\sim\IPq$ with support $(\alpha,\omega)$
and $g\dvtx (\alpha,\omega)\to\R$ is
an absolutely continuous function with a.s. derivative $g'$ such that
$\E q(X)(g'(X))^2<\infty$ then
\mbox{$\E g^2(X)<\infty$}.
\end{lem}
\begin{pf}
Observe that $g^2(X)\leq2g^2(\mu)+2(g(X)-g(\mu))^2$.
Since
$\mu\in(\alpha,\omega)$,
\begin{eqnarray*}
\E\bigl(g(X)-g(\mu)\bigr)^2 &=&\int_{\alpha}^{\mu}f(x)
\biggl(\int_{x}^{\mu}g'(t)\,\ud{t}
\biggr)^2 \,\ud{x} +\int_{\mu}^{\omega}f(x)
\biggl(\int_{\mu}^{x}g'(t)\,\ud{t}
\biggr)^2\,\ud{x}
\\
&\leq&\int_{\alpha}^{\mu} f(x) (\mu-x) \int
_{x}^{\mu} \bigl(g'(t)
\bigr)^2\,\ud {t}\,\ud{x} + \int_{\mu}^{\omega}
f(x) (x-\mu) \int_{\mu}^{x} \bigl(g'(t)
\bigr)^2\,\ud {t}\,\ud{x}
\\
&=&\E q(X) \bigl(g'(X)\bigr)^2,
\end{eqnarray*}
by the Cauchy--Schwarz inequality and Tonelli's theorem; cf. Lemma~3.1 in \cite{PP}.
\end{pf}
%
\begin{cor}\label{cor.Eq^n g^(n)}
If $X\sim\IPq$, $\E|X|^{2n-1}<\infty$ and $g\in\CH^{n}(X)$
for some fixed $n\in\{1,2,\ldots\}$ then
$\E q^{i}(X)(g^{(i)}(X))^2 <\infty$
for all $i\in\{0,1,\ldots,n\}$. In particular,
$\Var g(X)<\infty$, that is, $g\in L^2(\R,X)$.
\end{cor}
\begin{pf}
According to Theorem~\ref{theo.star}, the assumptions on $X$
enable us to define the random variables $X_k$ with densities
\[
f_k(x)=\frac{q^k(x)f(x)}{\E q^k(X)},\qquad  \alpha<x<\omega, k=0,1,\ldots,n-1,
\]
where $(\alpha,\omega)$
is the support of $X$ (and of each $X_k$). If
$q(x)=\delta x^2+\beta x+\gamma$
is the quadratic of $X,$ then
$X_k\sim\mbox{IP}(\mu_k;q_k)$
with mean $\mu_k$ and quadratic $q_k$
given by
\[
\mu_k=\frac{\mu+k\beta}{1-2k\delta}, \qquad  q_k(x)=
\frac{\delta x^2
+\beta x+\gamma}{1-2k\delta} =\delta_k x^2+\beta_k x+
\gamma_k, \qquad k=0,1,\ldots,n-1.
\]
Set $\widetilde{g}=g^{(n-1)}$,
$\widetilde{\mu}=\mu_{n-1}$,
$\widetilde{q}=q_{n-1}$,
$\widetilde{X}=X_{n-1}$
and observe that
$\widetilde{X}\sim\IPtilde$
and
\[
\E \widetilde{q} (\widetilde{X}) \bigl(\widetilde{g}'(
\widetilde{X})\bigr)^2 =\frac{\E q^{n}(X)(g^{(n)}(X))^2}{(1-(2n-2)\delta)
\E q^{n-1}(X)}<\infty,
\]
because $g\in\CH^n(X)$ so that the numerator is finite.
[In view of
Lemma~\ref{lem.moments},
$\E|X|^{2n-1}<\infty$ implies
the inequality $(2n-2)\delta<1$; moreover,
$\deg(q^{n-1})\leq2n-2$ shows that
$0<\E q^{n-1}(X)<\infty$.]
An application of Lemma~\ref{lem.deriv.l2} to
$\widetilde{g}$,
$\widetilde{X}$
shows that
$\E\widetilde{g}^2(\widetilde{X})<\infty$, and thus,
%
\[
\E q^{n-1}(X) \bigl(g^{(n-1)}(X)\bigr)^2= \E
\widetilde{g}^2(\widetilde{X})\E q^{n-1}(X)<\infty.
\]
Hence, $g\in\CH^{n-1}(X)$. Continuing inductively the result follows.
\end{pf}

Turn now to the case where $X\sim\IP$ with
$\delta\leq0$.
It follows that all moments exist and, moreover,
the moment generating function
of $X$ is finite in a neighborhood of zero
(see~\cite{AP2}, Table~2.1, types 1--3).
Then, it is well known that the orthonormalized
polynomial system
$\{\phi_k\}_{k=0}^{\infty}$, given by (\ref{orthonormal})
(with $n=\infty$),
is complete in $L^2(\R,X)$; see, for example, \cite{BergCrist,APP2};
see also Remark~\ref{rem.complete}, below.
Consider a function $g\in\CH^n(X)$ for some fixed
$n\in\{1,2,\ldots\}$.
Since $\CH^n(X)\subseteq L^2(\R,X)$, $g$
can be expanded as
\begin{equation}\label{Fourier.g}
g(x)\sim\sum_{k=0}^{\infty}
\alpha_k \phi_k(x),
\end{equation}
where $\alpha_k=\E\phi_k(X)g(X)$ are the Fourier
coefficients of $g$.
The series converges
in the norm of $L^2(\R,X)$, that is,
$\E[g(X)-\sum_{k=0}^N \alpha_k \phi_k(X)]^2\to0$
as $N\to\infty$. Parseval's identity shows that
\begin{equation}\label{Parseval.g}
\Var g(X)=\sum_{k=1}^{\infty}
\alpha_k^2, \qquad g\in L^2(\R,X).
\end{equation}
On the other hand, since $g\in\CH^n(X)$,
(\ref{Fourier2}) yields the expression
\[
\alpha_k=\frac{\E q^k(X) g^{(k)}(X)}{\sqrt{k!c_k(\delta)
\E q^k(X)}} \qquad \mbox{for } k=1,2,\ldots,n,
\]
where $c_k(\delta)=\prod_{j=k-1}^{2k-2}(1-j\delta)$,
see (\ref{lead}),
and $\E q^k(X)$
is given explicitly in (\ref{Eq^k}).
Thus, in the particular case where $g\in\CH^n(X)$,
(\ref{Parseval.g}) produces the equivalent formula
\begin{equation}\label{Parseval.g.2}
\Var g(X)=\sum_{k=1}^{n}
\frac{\E^2 q^k(X) g^{(k)}(X)}{k!c_k(\delta)\E q^k(X)} +\sum_{k=n+1}^\infty
\alpha_k^2,\qquad  g\in\CH^n(X).
\end{equation}

Now, consider the following heuristic derivation:
Formally, we differentiate term by term
($n$ times) the series (\ref{Fourier.g}) to get,
in view of Theorem~\ref{theor.derivatives},
the expansion
\begin{equation}\label{Fourier.g.n}
g^{(n)}(x)\sim\sum_{k=0}^{\infty}
\alpha_{k+n} \phi_{k+n}^{(n)}(x) =\sum
_{k=0}^{\infty} \nu_k^{(n)}
\alpha_{k+n} \phi_{k,n}(x).
\end{equation}
Let $\lead(P)$ be the leading coefficient of a polynomial $P$. The
constants $\nu_k^{(n)}=\nu_k^{(n)}(\mu;q)$
are given by (\ref{orthonormal.constants}) and
$\{\phi_{k,n}(x)\}_{k=0}^{\infty}$
(with $\lead(\phi_{k,n})>0$)
is the
orthonormal
polynomial system
corresponding to
$X_n$ with density
$f_n=q^n f/\E q^n(X)$; $\phi_{k,n}$ is a (positive)
scalar multiple of the polynomial $P_{k,n}$
given in (\ref{polyn.orthogonal.gen.star}).
Now, if the expansion (\ref{Fourier.g.n}) was indeed correct
in the $L^2(\R,X_n)$-sense, then
the completeness of the system $\{\phi_{k,n}\}_{k=0}^{\infty}$
in $L^2(\R,X_n)$ would result to the corresponding Parseval
identity:
\begin{equation}\label{Parseval.g.n}
\frac{\E q^n(X) (g^{(n)}(X) )^2}{\E q^n(X)} =\E \bigl(g^{(n)}(X_n)
\bigr)^2 =\sum_{k=0}^{\infty} \bigl(
\nu_k^{(n)} \bigr)^2 \alpha_{k+n}^2,\qquad
g\in\CH^n(X).
\end{equation}
Finally, from (\ref{orthonormal.constants}) we have
\[
\bigl(\nu_k^{(n)}\bigr)^2=\frac{(k+n)!}{k!\E
q^n(X)}
\prod_{j=k+n-1}^{k+2n-2}(1-j\delta).
\]
A combination of the last equation with (\ref{Parseval.g.n})
yields the identity
\begin{eqnarray}\label{Eq^n g^(n)}
\E q^n(X)
\bigl(g^{(n)}(X)\bigr)^2 &=& \sum_{k=0}^{\infty}
\frac{(k+n)!\prod_{j=k+n-1}^{k+2n-2}(1-j\delta)}{k!}
\alpha_{k+n}^2\nonumber\\[-8pt]\\[-8pt]
 &=& \sum
_{k=n}^{\infty} \frac{k!\prod_{j=k-1}^{k+n-2}(1-j\delta)}{(k-n)!} \alpha_{k}^2.\nonumber
\end{eqnarray}
This must be correct for all $g\in\CH^n(X)$, provided that
expansion (\ref{Fourier.g.n}) is valid. However, the above
arguments are heuristic; they are not
sufficient even to conclude convergence of the series
(\ref{Eq^n g^(n)}) or (\ref{Parseval.g.n}).
Notice that the same technicality appeared in Chernoff's \cite{Cher}
proof, although in this case the polynomials are the well-known
Hermite polynomials (with derivatives again Hermite, i.e., orthogonal
to the same
weight function, the normal density).
Chernoff overcame this difficulty by applying Weierstrass
(uniform) approximations to $g$ in compact intervals.

In the sequel, we shall make the above arguments rigorous
by applying a different technique, in the spirit of
Sturm--Liouville theory. In fact, we shall show more,
namely,
that an initial segment of the Fourier
coefficients for the $n$th derivative of $g$,
suggested by (\ref{Fourier.g.n}),
can be derived for any $X\sim\IP$
having a sufficient number of moments. This result holds
even if $\delta>0$, noting that if $\delta>0$ then $X$
possesses
only a finite number of moments.
Specifically, the following result,
which may have some interest in itself,
holds true.
%
\begin{lem}\label{lem.giwrgis}
Assume that $X$ has density $f$, support
$(\alpha,\omega)$,
$X\sim\IP$ and $\E|X|^{2N}<\infty$ for some $N\geq1$,
that is, $\delta<\frac{1}{2N-1}$.
Let $\{\phi_k\}_{k=0}^N\subseteq L^{2}(\R,X)$
be the orthonormal polynomial system associated with $X$
(where, to be specific, assume that $\lead(\phi_k)>0$).
Then,
for every $x\in(\alpha,\omega)$,
%
\begin{eqnarray}\label{der.poly}
q(x)f(x)\phi_k'(x)&=&-
\lambda_k(\delta) \int_\alpha^x
\phi_k(y)f(y)\,\ud{y}\nonumber\\[-8pt]\\[-8pt]
&=&\lambda_k(\delta) \int
^\omega_x \phi_k(y)f(y)\,\ud{y},\qquad
k=1,2,\ldots,N,\nonumber
\end{eqnarray}
where $\lambda_k(\delta):=k(1-(k-1)\delta)$.
Moreover, if $g\in\CH^n(X)$ for some $n\in\{1,2,\ldots,N\}$
then
\begin{equation}\label{Fourier.g.n.2}
\E\phi_{k,n}(X_n)g^{(n)}(X_n)=
\nu_k^{(n)} \E\phi_{k+n}(X)g(X), \qquad k=0,1,\ldots,N-n,
\end{equation}
where $X_n$ has density $f_n=q^n f/\E q^n(X)$,
\[
\nu_k^{(n)}=\sqrt{\frac{(k+n)!}{k!}
\frac{\prod_{j=k+n-1}^{k+2n-2}(1-j\delta)}{\E q^n(X)}}
\]
is given by (\textup{\ref{orthonormal.constants}}) and
$\{\phi_{k,n}\}_{k=0}^{N-n}\subseteq L^2(\R,X_n)$
is the orthonormal polynomial system corresponding to $X_n$,
with $\lead(\phi_{k,n})>0$.
\end{lem}
\begin{pf}
From (\ref{qua1}) it follows that
\[
\frac{f'(x)}{f(x)}=\frac{\mu-x-q'(x)}{q(x)} =\frac{-(1+2\delta)x+(\mu-\beta)}{\delta x^2+\beta x+\gamma}, \qquad \alpha<x<\omega.
\]
Consider the polynomials $P_k$ defined in (\ref{Rodrigues}).
By (\ref{orthonormal}), each $\phi_k$
is a scalar multiple of the Rodrigues-type polynomial
$h_k=D^k[q^k f]/f=(-1)^k P_k$.
%
Hence, Theorem~1 of Diaconis and Zabell \cite{DZ} (see,
also, equation (4.4) in \cite{AP2})
implies that
\begin{equation}\label{sl2}
\bigl[q(x)f(x)\phi_k'(x)
\bigr]'=-\lambda_k(\delta) \phi_k(x)f(x),\qquad
\alpha<x<\omega,
k=1,2,\ldots,N.
\end{equation}
Fix $t$ and $x$ with $\alpha<t<x<\omega$
and integrate (\ref{sl2}) over
the interval $[t,x]$ to get
\[
-\lambda_k(\delta)\int_{t}^x
\phi_k(y)f(y)\,\ud{y}=q(x)f(x)\phi_k'(x)-q(t)f(t)
\phi_k'(t);
\]
thus, taking limits as $t\searrow\alpha$ we see that the l.h.s.
converges to
$-\lambda_k(\delta)\int_{\alpha}^x \phi_k(y)f(y)\,\ud{y}$,
by dominated convergence, while the r.h.s. tends to $q(x)f(x)\phi_k'(x)$ because, by Lemma~\ref{lem.limits2}, $\lim_{t\searrow\alpha} q(t)f(t)h(t)=0$
for any polynomial $h$ with $\deg(h)\leq2N-1$. This verifies the
first equality in (\ref{der.poly}), while the second one
is obvious since $\E\phi_k(X)=0$ (because
$\phi_k$ is orthogonal to $\phi_0\equiv1$).

Fix now an integer $k\in\{0,1,\ldots,N-1\}$.
Observing that
$\deg(q(x)x^{2k})\leq2k+2\leq2N$ we have
$\E(X_1^{k})^2=\E q(X)X^{2k}/\E q(X)<\infty$. Thus,
the Rodrigues-type polynomial $P_{k,1}$
(see (\ref{polyn.orthogonal.gen.star}) with $m=1$)
belongs to $L^2(\R,X_1)$.
By Corollary~\ref{cor.Eq^n g^(n)}, $\E(g'(X_1))^2$ is also finite.
Indeed,
$n\leq N$ implies that $\E|X|^{2n-1}<\infty$ so that
$g\in\CH^n(X)\subseteq\CH^1(X)$ and, therefore,
\[
\E\bigl(g'(X_1)\bigr)^2=\frac{1}{\E q(X)}
\E q(X) \bigl(g'(X)\bigr)^2<\infty.
\]
%
Hence,
the Fourier coefficient
of $g'$ with respect to $\phi_{k,1}$, $\E\phi_{k,1}(X_1)g'(X_1)$,
is well-defined (and finite):
\[
\E^2\bigl|\phi_{k,1}(X_1)g'(X_1)\bigr|
\leq \E\bigl(\phi_{k,1}(X_1)\bigr)^2 \E
\bigl(g'(X_1)\bigr)^2=\E
\bigl(g'(X_1)\bigr)^2<\infty.
\]
Let
$\rho_1<\rho_2<\cdots<\rho_m$
be the distinct roots of $\phi_{k+1}$ that lie into the interval
$(\alpha,\omega)$. Clearly, $1\leq m\leq k+1$ because
$\E\phi_{k+1}(X)=0$ and $\deg(\phi_{k+1})=k+1$.
Fix now a number
$\rho\in[\rho_1,\rho_m]\subseteq(\alpha,\omega)$.
From
(\ref{eq.paragwgos}),
we see that $\phi_{k,1}(x)=\phi'_{k+1}(x)/\nu_k^{(1)}$
where $\nu_k^{(1)}=\sqrt{(k+1)(1-k\delta)/\E q(X)}$.
Therefore, using (\ref{der.poly}), we have
\begin{eqnarray*}
\E\phi_{k,1}(X_1)g'(X_1) &=&
\frac{1}{\E q(X)} \int_{\alpha}^{\omega}
g'(x) q(x)f(x)\phi_{k,1}(x) \,\ud{x}
\\
&=& \frac{1}{\nu_k^{(1)}\E q(X)} \int_{\alpha}^{\omega}
g'(x) q(x)f(x)\phi'_{k+1}(x) \,\ud{x}
\\
&=& \frac{-\lambda_{k+1}(\delta)}{\nu_k^{(1)}\E q(X)} \int_{\alpha
}^{\rho}
g'(x) \int_{\alpha}^{x} f(y)
\phi_{k+1}(y)\,\ud{y} \,\ud{x}
\\
&&{} +\frac{\lambda_{k+1}(\delta)}{\nu_k^{(1)}\E q(X)} \int_{\rho
}^{\omega}
g'(x) \int_{x}^{\omega} f(y)
\phi_{k+1}(y)\,\ud{y} \,\ud{x}.
\end{eqnarray*}
Observing that
\[
\frac{\lambda_{k+1}(\delta)}{\nu_k^{(1)}\E q(X)} =\frac{(k+1)(1-k\delta)}{\E q(X)\sqrt{(k+1)(1-k\delta)/\E q(X)}} 
=\nu_k^{(1)},
\]
the preceding equation
can be rewritten as
\begin{equation}\label{eq.2.10}
\E\phi_{k,1}(X_1)g'(X_1)
=\nu_k^{(1)} (I_2-I_1),
\end{equation}
where
\begin{equation}\label{eq.2.11}
I_1:= \int_{\alpha}^{\rho}
g'(x) \int_{\alpha}^{x} f(y)
\phi_{k+1}(y)\,\ud{y} \,\ud{x},\qquad  I_2:= \int_{\rho}^{\omega}
g'(x) \int_{x}^{\omega} f(y)
\phi_{k+1}(y)\,\ud{y} \,\ud{x}.
\end{equation}
Now, we wish to change the order of integration to both integrals
$I_1$ and $I_2$. To this end, for $I_2$ it suffices to show that
\begin{equation}\label{eq.2.12}
I_2^*:= \int_{\rho}^{\omega}
\bigl|g'(x)\bigr| \int_{x}^{\omega} f(y)\bigl|
\phi_{k+1}(y)\bigr|\,\ud{y} \,\ud{x}<\infty.
\end{equation}
Similarly, for $I_1$
it suffices to show that
$I_1^*:= \int_{\alpha}^{\rho} |g'(x)|
\int_{\alpha}^{x} f(y)|\phi_{k+1}(y)|\,\ud{y} \,\ud{x}<\infty$.
We now proceed to verify (\ref{eq.2.12}).
Write $I_2^*=I_{21}^*+I_{22}^*$ where
\begin{eqnarray*}
I_{21}^*&:=& \int_{\rho}^{\rho_m}
\bigl|g'(x)\bigr| \int_{x}^{\omega} f(y)\bigl|
\phi_{k+1}(y)\bigr|\,\ud{y} \,\ud{x}, \\
 I_{22}^*&:=& \int
_{\rho_m}^{\omega} \bigl|g'(x)\bigr| \int
_{x}^{\omega} f(y)\bigl|\phi_{k+1}(y)\bigr|\,\ud{y}
\,\ud{x}.
\end{eqnarray*}
Since the polynomial
$\phi_{k+1}$ does not change sign in the interval $(\rho_m,\omega)$,
we can define the constant $\pi$ as
\[
\pi:=\sign\bigl(\phi_{k+1}(x)\bigr)\in\{-1,1\},\qquad  \rho_m<x<
\omega.
\]
Then, $\pi\phi_{k+1}(x)=|\phi_{k+1}(x)|$ holds
for all $x\in(\rho_m,\omega)$ and
from (\ref{der.poly}) we get
\begin{eqnarray*}
I_{22}^* &=& \pi\int_{\rho_m}^{\omega}
\bigl|g'(x)\bigr| \int_{x}^{\omega} f(y)
\phi_{k+1}(y) \,\ud{y} \,\ud{x} =\frac{\pi}{\lambda_{k+1}(\delta)} \int_{\rho_m}^{\omega}
\bigl|g'(x)\bigr| q(x)f(x)\phi_{k+1}'(x)\,\ud{x}
\\
&\leq& \frac{1}{\lambda_{k+1}(\delta)} \int_{\rho_m}^{\omega}
\bigl|g'(x)\bigr| q(x)f(x)\bigl|\phi_{k+1}'(x)\bigr|\,\ud{x}
\\
&\leq& \frac{1}{\lambda_{k+1}(\delta)} \int_{\alpha}^{\omega}
\bigl|g'(x)\bigr| q(x)f(x)\bigl|\phi_{k+1}'(x)\bigr|\,\ud{x} =
\frac{1}{\lambda_{k+1}(\delta)} \E q(X) \bigl|\phi_{k+1}'(X)g'(X)\bigr|
\\
&=& \frac{\nu_k^{(1)}}{\lambda_{k+1}(\delta)} \E q(X) \bigl|\phi_{k,1}(X)g'(X)\bigr|
=\frac{1}{\nu_k^{(1)}}\E\bigl|\phi_{k,1}(X_1)g'(X_1)\bigr|<
\infty.
\end{eqnarray*}
This shows that $I_{22}^*<\infty$. On the other hand, the function
$x\mapsto q(x)f(x)$ is strictly positive and continuous for
$x$ in the compact interval
$[\rho,\rho_m]\subseteq(\alpha,\omega)$, so that,
$\theta:=\min\{q(x)f(x):{\rho\leq x\leq\rho_m}\}>0$.
Then, from the fact that $g\in\CH^1(X)$, we get
\begin{eqnarray*}
\int_{\rho}^{\rho_m} \bigl|g'(x)\bigr|\, \ud{x} &\leq&
\frac{1}{\theta}\int_{\rho}^{\rho_m}
q(x)f(x)\bigl|g'(x)\bigr|\,\ud{x} \leq\frac{1}{\theta}\E q(X)\bigl|g'(X)\bigr|
\\
&\leq& \frac{1}{\theta}\sqrt{\E q(X)\E q(X) \bigl(g'(X)
\bigr)^2} <\infty.
\end{eqnarray*}
Moreover, for any $u_1,u_2$ with $\alpha\leq u_1\leq u_2\leq\omega$
it is readily seen that
\[
\int_{u_1}^{u_2} \bigl|\phi_{k+1}(y)\bigr| f(y)\,\ud{y}
\leq \int_{\alpha}^{\omega} \bigl|\phi_{k+1}(y)\bigr| f(y)
\,\ud{y}= \E\bigl|\phi_{k+1}(X)\bigr|:=M_{k+1}<\infty.
\]
Combining the above, we conclude that
\[
I_{21}^*= \int_{\rho}^{\rho_m}
\bigl|g'(x)\bigr| \int_{x}^{\omega} f(y)\bigl|
\phi_{k+1}(y)\bigr|\,\ud{y}\, \ud{x} \leq M_{k+1}\int
_{\rho}^{\rho_m} \bigl|g'(x)\bigr| \,\ud{x} 
<\infty.
\]
Therefore, $I_2^*=I_{21}^*+I_{22}^*<\infty$
and (\ref{eq.2.12}) follows. Using similar arguments it
is shown that $I_1^*<\infty$. Thus,
we can indeed interchange
the order of integration to both integrals
$I_1$ and $I_2$ of (\ref{eq.2.11}). It follows that
\begin{eqnarray*}
I_2 &=& \int_{\rho}^{\omega} f(y)
\phi_{k+1}(y) \int_{\rho}^{y}
g'(x) \,\ud{x} \,\ud{y} \\
&=& \int_{\rho}^{\omega}
f(y)\phi_{k+1}(y)g(y)\,\ud{y} -g(\rho) \int_{\rho}^{\omega}
f(y)\phi_{k+1}(y)\,\ud{y}
\end{eqnarray*}
and, similarly,
\[
I_1 =g(\rho)\int_{\alpha}^{\rho} f(y)
\phi_{k+1}(y)\,\ud{y} -\int_{\alpha}^{\rho} f(y)
\phi_{k+1}(y)g(y)\,\ud{y}.
\]
Taking into account the fact that
$\int_{\alpha}^{\omega} f(y)\phi_{k+1}(y)\,\ud{y}=\E
\phi_{k+1}(X)=0$, we get
\[
I_2-I_1=\int_{\alpha}^{\omega}
f(y)\phi_{k+1}(y)g(y)\,\ud{y} -g(\rho)\int_{\alpha}^{\omega}
f(y)\phi_{k+1}(y)\,\ud{y}= \E\phi_{k+1}(X)g(X).
\]
Finally, from (\ref{eq.2.10}), we conclude that
\begin{equation}\label{giwrgis}
\E\phi_{k,1}(X_1)g'(X_1)
= 
\sqrt{\frac{(k+1)(1-k\delta)}{\E q(X)}} \E\phi_{k+1}(X)g(X),\qquad
k=0,1,\ldots,N-1.
\end{equation}
So far we have shown that $g\in\CH^n(X)$
and $\E|X|^{2N}<\infty$ for some $N\geq n$
implies that $g\in\CH^1(X)$ and
(\ref{giwrgis}) is fulfilled.
Assume now that
for some $i\in\{1,2,\ldots,n-1\}$ we
have shown that
$g\in\CH^i(X)$ and that
for every $k\in\{0,1,\ldots,N-i\}$,
\begin{equation}\label{giwrgis2}
\E\phi_{k,i}(X_i)g^{(i)}(X_i)
= 
\sqrt{\frac{(k+i)!}{k!}
\frac{\prod_{j=k+i-1}^{k+2i-2}(1-j\delta)}{\E q^i(X)}} \E\phi_{k+i}(X)g(X). 
\end{equation}
Clearly, we can apply (\ref{giwrgis})
for $g=g^{(i)}$, $X=X_i$ and for $k=0,1,\ldots,\widetilde{N}-1$,
provided that $\E|X_i|^{2\widetilde{N}}<\infty$.
Observing that $\E|X_i|^{2\widetilde{N}}
=\frac{\E q^i(X)|X|^{2\widetilde{N}}}{\E q^i(X)}$
it follows that $\widetilde{N}=N-i$ is a suitable choice.
Therefore, for $k=0,1,\ldots,N-i-1$, (\ref{giwrgis}) yields
\[
\E\phi_{k,i+1}(X_{i+1})g^{(i+1)}(X_{i+1}) =
\sqrt{\frac{(k+1)(1-k\delta_i)}{\E q_i(X_i)}} \E\phi_{k+1,i}(X_i)g^{(i)}(X_i),
\]
where
$\delta_i=\frac{\delta}{1-2i\delta}$, $q_i(x)=\frac
{q(x)}{1-2i\delta}$
(see Theorem~\ref{theo.star}) and, thus,
\[
\E q_i(X_i)=\frac{\E q(X_i)}{1-2i\delta} =\frac{\E q^{i+1}(X)}{(1-2i\delta)\E q^{i}(X)}.
\]
Finally, calculating $\E\phi_{k+1,i}(X_i)g^{(i)}(X_i)$
from (\ref{giwrgis2}) (for $k=0,1,\ldots,N-i-1$)
we see that
\begin{eqnarray*}
&&\E\phi_{k,i+1}(X_{i+1})g^{(i+1)}(X_{i+1})
\\
&&\quad  = \sqrt{\frac{(k+1)(1-k\delta/(1-2i\delta))} {
\E q^{i+1}(X)/((1-2i\delta)\E q^{i}(X))}} \sqrt{\frac{(k+i+1)!}{(k+1)!}
\frac{\prod_{j=k+i}^{k+2i-1}(1-j\delta)}{\E q^i(X)}} \E\phi_{k+i+1}(X)g(X)
\\
&&\quad  = \sqrt{\frac{(k+i+1)!}{k!}\frac{\prod_{j=k+i}^{k+2i-1}(1-j\delta)} {
\E q^{i+1}(X)}} \E\phi_{k+i+1}(X)g(X),\qquad
k=0,1,\ldots,N-i-1,
\end{eqnarray*}
which verifies the inductional step and shows that
(\ref{giwrgis2}) holds for all $i\in\{1,2,\ldots,n\}$.
Letting $i=n$ in (\ref{giwrgis2}) completes the proof.
\end{pf}

\section{The strengthened inequality}\label{sec3}

In the present section, we
assume that $X\sim\IP$ with
$\delta\leq0$. The well-known
Normal, Gamma and Beta random variables and their
affine transformations are of this form -- see \cite{AP2},
Table~2.1.
In this case the orthonormal polynomial system
$\{\phi_k\}_{k=0}^{\infty}$ is complete in
$L^2(\R,X)$ and, therefore, the following
result holds.
%
\begin{lem}\label{lem.3.1}
If $X\sim\IP$ with $\delta\leq0$,
then
\begin{equation}\label{var.complete}
\Var g(X)=\sum_{k=1}^\infty
\alpha_k^2 \qquad \mbox{for any } g\in L^2(\R,X),
\end{equation}
where
\begin{equation}\label{Fourier.coefficients}
\alpha_k=\E\phi_k(X)g(X),\qquad
k=0,1,2,\ldots,
\end{equation}
are the Fourier
coefficients of $g$ with respect to the
orthonormal polynomial system
$\{\phi_k\}_{k=0}^{\infty}$. If, furthermore,
$g\in\CH^n(X)$ for some $n\in\{1,2,\ldots\}$,
then
\begin{equation}\label{Fourier.g.k}
\alpha_k=\E\phi_k(X) g(X)=
\frac{\E q^k(X) g^{(k)}(X)}{\sqrt{k!\E q^k(X)
\prod_{j=k-1}^{2k-2}(1-j\delta)}}, \qquad k=1,2,\ldots,n
\end{equation}
and
\begin{equation}\label{E q^n g^(n).complete}
\E q^{n}(X)
\bigl(g^{(n)}(X) \bigr)^2 =\sum_{k=n}^{\infty}
\frac{k!\prod_{j=k-1}^{k+n-2}(1-j\delta)}{(k-n)!} \alpha_{k}^2,
\end{equation}
with $\alpha_k$ given by (\ref{Fourier.coefficients}).
\end{lem}
\begin{pf}
(\ref{var.complete}) is the
well-known Parseval's identity.
Also, if $g\in\CH^n(X)$ then,
by Corollary~\ref{cor.Eq^n g^(n)},
$g\in\CH^k(X)$ for all
$k \in\{0,1,\ldots,n\}$.
Therefore, the Cauchy--Schwarz inequality
shows that $\E q^k(X)|g^{(k)}(X)|\leq
\E q^k(X) \E q^{k}(X)(g^{(k)}(X))^2<\infty$. Hence,
(\ref{Fourier.g.k}) follows from (\ref{cov.id}) --
see Theorem~\ref{theo.cov} --
and the fact that the
polynomials
$P_k(x):=(-1)^k D^k[q^k(x)f(x)]/f(x)$
are related to
$\phi_k$ by
$P_k(x)=\phi_k(x)\sqrt{k!\E q^k(X)
\prod_{j=k-1}^{2k-2}(1-j\delta)}$ for all
$k\in\{1,2,\ldots\}$.
Moreover, by Lemma~\ref{lem.giwrgis} we have that
for any
$g\in\CH^n(X)$, the Fourier coefficients
$\alpha_k=\E\phi_k(X)g(X)$
(of $g$ with respect to $X$) and the Fourier coefficients
$\alpha_k^{(n)}:=\E\phi_{k,n}(X_n) g^{(n)}(X_n)$ of $g^{(n)}$
with respect to $X_n$ are related through
\[
\alpha_k^{(n)}= \sqrt{
\frac{(k+n)!}{k!} \frac{\prod_{j=k+n-1}^{k+2n-2}(1-j\delta)}{\E q^n(X)}} \alpha_{k+n},\qquad  k=0,1,2,\ldots,
\]
where $\E q^n(X)$ is given explicitly by (\ref{Eq^k}).
Finally,
Theorem~\ref{theo.star} asserts that
\[
X_n\sim\mbox{IP}(\mu_n;\delta_n,
\beta_n,\gamma_n) \qquad \mbox{with } \delta_n=
\frac{\delta}{1-2n\delta}\leq0.
\]
Hence, $\delta_n\leq0$
guarantees that
the corresponding orthonormal polynomial system
$\{\phi_{k,n}\}_{k=0}^{\infty}$ is complete in
$L^2(\R,X_n)$. Since $g\in\CH^n(X)$,
$g^{(n)}\in L^2(\R,X_n)$ and, by Parseval's identity,
\[
\E\bigl(g^{(n)}(X_n)\bigr)^2=\sum
_{k=0}^{\infty} \bigl(\alpha_k^{(n)}
\bigr)^2 =\frac{1}{\E q^n(X)}\sum_{k=0}^{\infty}
\frac{(k+n)!
\prod_{j=k+n-1}^{k+2n-2}(1-j\delta)}{k!} \alpha_{k+n}^2
\]
(thus, the series converges).
Observing that
\[
\E\bigl(g^{(n)}(X_n)\bigr)^2=\frac{1}{\E
q^n(X)}
\E q^n(X) \bigl(g^{(n)}(X)\bigr)^2,
\]
(\ref{E q^n g^(n).complete})
is deduced and the proof is complete.
\end{pf}

We are now in a position to state and prove the main result of the paper.
%
\begin{theo}\label{theo.main}
If $X\sim\IP$ with $\delta\leq0$
and if $g\in\CH^n(X)$ for some $n\in\{1,2,\ldots\}$
then
%
\begin{eqnarray}\label{main}
\Var g(X)&\leq& \sum_{k=1}^{n}
\frac{\E^2 q^k(X) g^{(k)}(X)} {
k!\E q^k(X)\prod_{j=k-1}^{2k-2}(1-j\delta)}\nonumber
\\[-8pt]\\[-8pt]
&&{} +\frac{\E q^n(X)(g^{(n)}(X))^2
-(1/\E q^n(X))\E^2 q^n(X) g^{(n)}(X)}{(n+1)!\prod_{j=n}^{2n-1}(1-j\delta)},\nonumber
%
\end{eqnarray}
with equality if and only if $g$ is a polynomial of
degree at most $n+1$.

In particular, if $\sigma^2=\Var X$ and $g$
is absolutely continuous with a.s. derivative $g'$
such that $\E q(X)(g'(X))^2<\infty$
(i.e., $g\in\CH^1(X)$) then
\begin{equation}
\label{main2} \Var g(X)\leq \biggl(1-\frac{1}{2(1-\delta)} \biggr)
\frac{1}{\sigma^2} \E^2 q(X) g'(X) +\frac{1}{2(1-\delta)}
\E q(X) \bigl(g'(X)\bigr)^2,
\end{equation}
with equality if and only if $g$ is a polynomial of degree at most
two.
\end{theo}
Three examples of (\ref{main2}) are as follows:
%
\begin{exam}\label{ex.normal}
If $X\sim N(\mu,\sigma^2)\equiv\operatorname{IP}(\mu;0,0,\sigma^2)$
then $\delta=0$,
$q(x)\equiv\sigma^2$ and we
obtain the inequality
%
\begin{equation}\label{normal.bound}
 \Var g(X)\leq\tfrac{1}{2} \sigma^2
\E^2 g'(X) +\tfrac{1}{2} \sigma^2\E
\bigl(g'(X)\bigr)^2, %
\end{equation}
in which the equality holds if and only if $g$ is a polynomial
of degree at most two.
Chernoff's upper bound, $\Var g(X)\leq\sigma^2\E(g'(X))^2$,
is strictly weaker than (\ref{normal.bound})
since, obviously, $\E^2 g'(X)\leq\E(g'(X))^2$,
and the equality holds if and only if $g$ is linear.
It should be noted that $\sigma^2\E^2 g'(X)$
is, actually, a lower bound for $\Var g(X)$;
see, for example, \cite{Cac}.
\end{exam}
%
\begin{exam}\label{ex.Gamma}
If $X\sim\Gamma(a,\lambda)\equiv\operatorname{IP}(a/\lambda
;0,1/\lambda,0)$
so that $f(x)=\lambda^a x^{a-1}\mathrm{e}^{-\lambda x}/\allowbreak \Gamma(a)$, $x>0$, then
$\delta=0$, $q(x)=x/\lambda$, $\sigma^2=a/\lambda^2$
and we
obtain the inequality
%
\begin{equation}
\label{Gamma.bound} \Var g(X)\leq\frac{1}{2a} \E^2 X
g'(X) +\frac{1}{2\lambda} \E X\bigl(g'(X)
\bigr)^2, %
\end{equation}
in which the equality holds if and only if $g$ is a polynomial
of degree at most two.
\end{exam}
%
\begin{exam}
\label{ex.Beta}
If $X\sim B(a,b)\equiv\mbox{IP}(\frac{a}{a+b};\frac{-1}{a+b},\frac
{1}{a+b},0)$
then $\delta=\frac{-1}{a+b}$, $q(x)=\frac{x(1-x)}{a+b}$,
$\sigma^2=\frac{ab}{(a+b)^2 (a+b+1)}$
and we
obtain the inequality
%
\begin{equation}\label{Beta.bound}
\Var g(X)\leq\frac{a+b+2}{2ab} \E^2 X(1-X)
g'(X) +\frac{1}{2(a+b+1)} \E X(1-X) \bigl(g'(X)
\bigr)^2, %
\end{equation}
in which the equality holds if and only if $g$ is a polynomial
of degree at most two. In the particular case
where $a=b=1$, $X=U$ is uniformly distributed over the interval
$(0,1)$ and (\ref{Beta.bound}) yields an improvement
of Polya's inequality (see, e.g., \cite{AB}),
\[
\int_0^1 g^2(x) \,\ud{x}- \biggl(
\int_0^1 g(x) \,\ud{x} \biggr)^2 \leq
\frac{1}{2}\int_0^1 x(1-x)
\bigl(g'(x)\bigr)^2 \,\ud{x}.
\]
Indeed, for $a=b=1$, (\ref{Beta.bound}) yields
\[
\int_0^1 g^2(x) \,\ud{x}- \biggl(
\int_0^1 g(x) \,\ud{x} \biggr)^2 \leq2
\biggl(\int_0^1 x(1-x)g'(x)\,
\ud{x} \biggr)^2 +\frac{1}{6} \int_0^1
x(1-x) \bigl(g'(x)\bigr)^2 \,\ud{x},
\]
and the upper bound is smaller than Polya's bound
because, by the Cauchy--Schwarz inequality,
\begin{eqnarray*}
\biggl(\int_0^1 x(1-x)g'(x)
\,\ud{x} \biggr)^2&\leq& \int_0^1 x(1-x)
\,\ud{x} \int_0^1 x(1-x) \bigl(g'(x)
\bigr)^2 \,\ud{x}\\
&=&\frac16\int_0^1
x(1-x) \bigl(g'(x)\bigr)^2 \,\ud{x}.
\end{eqnarray*}
\end{exam}
%
\begin{rem}\label{rem.Cacoullos}
In \cite{CP1,John,PP} it was shown
that $\Var g(X)\leq\E q(X)(g'(X))^2$; the equality
in this Chernoff-type variance bound is attained only
by linear functions $g$. Also, in \cite{Cac,John,PP,CP2}
it was shown that
$\Var g(X)\geq\frac{1}{\sigma^2}\E^2 q(X)g'(X)$,
in which the equality characterizes again the
linear functions.
We observe that the upper bound in (\ref{main2})
is a convex combination of the preceding
lower and upper bounds and, thus, smaller than
the Chernoff-type upper bound, $\E q(X)(g'(X))^2$.
Also, the last term in the upper bound
(\ref{main})
can be rewritten as
\[
\frac{\E q^n(X)(g^{(n)}(X))^2-(1/\E q^n(X))\E^2 q^n(X)
g^{(n)}(X)} {
(n+1)!\prod_{j=n}^{2n-1}(1-j\delta)} =\frac{\E q^n(X)}{(n+1)!\prod_{j=n}^{2n-1}(1-j\delta)} \Var g^{(n)}(X_n).
\]
Thus, we can apply the Chernoff-type upper bound
to $\Var g^{(n)}(X_n)$, provided that $g^{(n)}\in\CH^1(X_n)$.
Recall that $g^{(n)}\in\CH^1(X_n)$ means that $g^{(n)}$
is absolutely continuous with a.s. derivative
$g^{(n+1)}$ such that $\E q_n(X_n)(g^{(n+1)}(X_n))^2<\infty$.
Since $X_n\sim f_n=q^n f/\E q^n(X)$, $\delta\leq0$
and $q_n(x)=q(x)/(1-2n\delta)$, the preceding requirement
is equivalent to
\[
\frac{1}{(1-2n\delta)\E q^n(X)} 
\E q^{n+1}(X) \bigl(g^{(n+1)}(X)
\bigr)^2<\infty;
\]
thus,
$g^{(n)}\in\CH^1(X_n)$ if and only if $g\in\CH^{n+1}(X)$.
Therefore, if $g\in\CH^{n+1}(X)$ then we have
\[
\Var g^{(n)}(X_n)\leq\E q_n(X_n)
\bigl(g^{(n+1)}(X_n)\bigr)^2 
=
\frac{\E q^{n+1}(X)(g^{(n+1)}(X))^2}{(1-2n\delta)
\E q^{n}(X)},
\]
with equality if and only if $g^{(n)}$ is linear, that is,
$g$ is a polynomial of degree at most $n+1$. The
preceding inequality shows that for any $g\in\CH^{n+1}(X)$,
\[
\frac{\E q^n(X)(g^{(n)}(X))^2-(1/\E q^n(X))\E^2 q^n(X)
g^{(n)}(X)} {
(n+1)!\prod_{j=n}^{2n-1}(1-j\delta)}\leq \frac{\E q^{n+1}(X)(g^{(n+1)}(X))^2}{(n+1)!\prod_{j=n}^{2n}(1-j\delta)}, %
\]
with equality only for polynomial $g$ of degree at most
$n+1$. Combining the upper bound in (\ref{main})
with the last displayed inequality, we
obtain the weaker bound
%
\begin{equation}\label{Chernoff.n}
\Var g(X)\leq\sum_{k=1}^{n-1}
\frac{\E^2 q^k(X) g^{(k)}(X)} {
k!\E q^k(X)\prod_{j=k-1}^{2k-2}(1-j\delta)} +\frac{\E q^{n}(X)
(g^{(n)}(X))^2}{n!\prod_{j=n-1}^{2n-2}(1-j\delta)},
\end{equation}
which holds for any $g\in\CH^n(X)$, and the equality is attained
if and only if $g$ is a polynomial of degree at most $n$.
For $n=1$ this is the Chernoff-type variance bound.
Also, for $X\sim B(a,b)$,
(\ref{Chernoff.n}) has been shown by Wei and Zhang \cite{WZ},
using Jacobi polynomials.
\end{rem}
\begin{pf*}{Proof of Theorem~\ref{theo.main}}
From (\ref{var.complete}) and (\ref{Fourier.g.k}),
\begin{equation}\label{ouf}
\Var g(X)-\sum_{k=1}^{n}
\frac{\E^2 q^k(X) g^{(k)}(X)}{k!\E
q^k(X)\prod_{j=k-1}^{2k-2}(1-j\delta)} = \alpha_{n+1}^2+\alpha_{n+2}^2+
\cdots,
\end{equation}
with $\alpha_k$ given by
(\ref{Fourier.coefficients}).
Also, from (\ref{Fourier.g.k}) with $k=n$,
\[
\frac{1}{\E q^n(X)}\E^2 q^n(X)
g^{(n)}(X) = n! \Biggl(\prod_{j=n-1}^{2n-2}(1-j
\delta) \Biggr) \alpha_n^2.
\]
Thus, in view of (\ref{E q^n g^(n).complete}),
\begin{eqnarray*}
&&\E q^n(X) \bigl(g^{(n)}(X)\bigr)^2-
\frac{1}{\E q^n(X)}\E^2 q^n(X) g^{(n)}(X)
\\
&&\quad =\sum_{k=n}^{\infty} \frac{k!\prod_{j=k-1}^{k+n-2}(1-j\delta)}{(k-n)!}
\alpha_{k}^2-n! \Biggl(\prod_{j=n-1}^{2n-2}(1-j
\delta) \Biggr) \alpha_n^2 = \sum
_{k=n+1}^{\infty} \frac{k!\prod_{j=k-1}^{k+n-2}(1-j\delta)}{(k-n)!}\alpha_{k}^2.
\end{eqnarray*}
Therefore,
\begin{eqnarray*}
& &\frac{\E q^n(X)(g^{(n)}(X))^2-(1/\E q^n(X))\E^2 q^n(X)
g^{(n)}(X)}{(n+1)!
\prod_{j=n}^{2n-1}(1-j\delta)}
\\
& &\quad =\sum_{k=n+1}^{\infty}\frac{k!\prod_{j=k-1}^{k+n-2}(1-j\delta)} {
(k-n)!(n+1)!\prod_{j=n}^{2n-1}(1-j\delta)}
\alpha_{k}^2 =\alpha_{n+1}^2+\sum
_{k=n+2}^\infty\lambda_k
\alpha_k^2,
\end{eqnarray*}
where
\[
\lambda_k:=\frac{1}{n+1}{k\choose n} \frac{\prod_{j=k-1}^{k+n-2}(1-j\delta)}{\prod_{j=n}^{2n-1}(1-j\delta)},
\qquad k=n+2,n+3,\ldots .
\]
The sequence $\{\lambda_k\}_{k=n+2}^{\infty}$
is nondecreasing in $k$. Indeed, since $\delta\leq0$, we have
\[
1\leq1-\delta\leq1-2\delta\leq1-3\delta\leq\cdots
\]
and thus,
$k\mapsto\prod_{j=k-1}^{k+n-2}(1-j\delta)$ is nondecreasing in $k$
and positive
(for each $k$ the product contains $n$ positive factors). Also,
\[
k\mapsto{k\choose n}
\]
is, obviously, positive and nondecreasing in $k$.
Thus, for every $k\geq n+2$,
\[
\lambda_k\geq \lambda_{n+2}= \biggl(1+\frac{n}{2}
\biggr) \biggl(1-\frac{n\delta}{1-n\delta} \biggr)>1,
\]
because $1+n/2>1$ and $1-n\delta/(1-n\delta)\geq1$
(since $\delta\leq0$). It follows that
\begin{equation}\label{ouf2}
\frac{\E q^n(X)(g^{(n)}(X))^2-(1/\E q^n(X))
\E^2 q^n(X) g^{(n)}(X)} {
(n+1)!\prod_{j=n}^{2n-1}(1-j\delta)}\geq \alpha_{n+1}^2+
\alpha_{n+2}^2+\cdots,
\end{equation}
with equality if and only if $\alpha_{n+2}=\alpha_{n+3}=\cdots=0$,
that is,
if and only if $g$ is a polynomial of degree at most $n+1$.
A combination of (\ref{ouf}) and (\ref{ouf2})
completes the proof.
\end{pf*}
%
\begin{rem}
\label{rem.question2}
The upper bound in (\ref{main})
is meaningful (it is nonnegative and makes sense)
even for
$0<\delta<\frac{1}{2n-1}$, in which case
$\E|X|^{2n}<\infty$. Also,
since
$x^{n+1}\in L^2(\R,X)$ if and only if $\delta<\frac{1}{2n+1}$,
it would be desirable to show the validity of (\ref{main})
at least when $0<\delta<\frac{1}{2n+1}$.
For example, we have tried, without success,
to prove (\ref{main2}) when
$0<\delta<\frac{1}{3}$. In contrast to the corresponding
Chernoff-type bound, which can be shown directly
(without Fourier expansions -- see, e.g.,
\cite{Chen}; cf. Lemma~\ref{lem.deriv.l2}, above),
it seems that the completeness of the corresponding
orthonormal polynomial
system in $L^2(\R,X)$ plays a crucial role
in proving (\ref{main2}).
\end{rem}

\begin{appendix}
\section*{Appendix} \label{app}
%
\begin{prop}[(\cite{AP2}, Proposition~2.1)]
\label{prop.a1}
Let $X\sim\IPq$ and set $(\alpha,\omega):=(\essinf(X),\break  \esssup(X))$.
Then, there is a version $f$ of the density of $X$ such
that
\begin{enumerate}[(iii)]
\item[(i)]
$f(x)$ is strictly positive for $x$ in $(\alpha,\omega)$ and zero
otherwise, that is,
$\{x\dvtx f(x)>0\}=(\alpha,\omega)$;
\item[(ii)]
$f\in C^{\infty}(\alpha,\omega)$, that is, $f$ has derivatives
of any order in $(\alpha,\omega)$;
\item[(iii)]
$X$ is a (usual) Pearson random variable supported in $(\alpha,\omega
)$, that is,
$f'(x)/f(x)=p_1(x)/q(x)$, $x\in(\alpha,\omega)$, where $p_1(x)=\mu-x-q'(x)$
is a polynomial of degree at most one;
\item[(iv)] $q(x)=\delta x^2+\beta x+\gamma>0$ for
all $x\in(\alpha,\omega)$;
\item[(v)]
if $\alpha>-\infty$ then $q(\alpha)=0$ and, similarly, if $\omega
<+\infty$
then $q(\omega)=0$;
\item[(vi)]
for any $\theta,c\in\R$ with $\theta\neq0$, the random
variable $\widetilde{X}:=\theta X+c\sim\IPtilde$
with $\widetilde{\mu}=\theta\mu+c$ and
$\widetilde{q}(x)=\theta^2 q((x-c)/\theta)$.
\end{enumerate}
\end{prop}
\setcounter{lem}{0}
\begin{lem}[(\cite{AP2}, Corollary~2.2)]\label{lem.moments}
Assume that $X\sim\IP$.
\begin{enumerate}[(ii)]
\item[(i)] If $\delta\leq0,$ then $\E|X|^\theta<\infty$
for any $\theta\in[0,\infty)$.
\item[(ii)] If $\delta>0,$ then $\E|X|^\theta<\infty$ for any
$\theta\in[0,1+1/\delta)$,
while $\E|X|^{1+1/\delta}=\infty$.
\end{enumerate}
\end{lem}
%
\begin{lem}[(\cite{AP2}, Lemma~2.1)]\label{lem.limits2}
If $X\sim\IP\equiv\IPq$ has support $(\alpha,\omega)$
and $\E|X|^n<\infty$ for some $n\geq1$ (equivalently, $\delta<1/(n-1)$),
then for any polynomial $Q_{n-1}$ of degree at most $n-1$,
\setcounter{equation}{0}
\begin{equation}\label{limits2}
\lim_{x\nearrow\omega}{q(x)f(x)Q_{n-1}(x)} =
\lim_{x\searrow\alpha}{q(x)f(x)Q_{n-1}(x)}=0.
\end{equation}
\end{lem}
%
%
\setcounter{theo}{0}
\begin{theo}[(\cite{Hild}, page 401; \cite{Beale2}, pages 99--100;
\cite{DZ}, page 295; \cite{AP2}, Theorem~4.1)]\label{theo.polynomials}
Assume that $f$ is the density of a random
variable $X\sim\IPq\equiv\IP$ with support
$(\alpha,\omega)$. Then, the functions
$P_k\dvtx (\alpha,\omega)\to\R$ with
\begin{equation}\label{Rodrigues}
P_k(x):=\frac{(-1)^k}{f(x)}\frac{\ud^k}{\ud{x}^k}
\bigl[q^{k}(x)f(x)\bigr], \qquad \alpha<x<\omega, k=0,1,2,\ldots
\end{equation}
are (Rodrigues-type)
polynomials with
\begin{equation}\label{lead}
\deg(P_k)\leq k \quad \mbox{and}\quad  \lead(P_k)=\prod
_{j=k-1}^{2k-2}(1-j\delta):=c_k(
\delta),\qquad  k=0,1,2,\ldots,
\end{equation}
where $\lead(P_k)$ is the coefficient of $x^k$ in
$P_k(x)$. Here $c_0(\delta):=1$, that is, an empty product
should be
treated as one.
\end{theo}

\begin{theo}[(\cite{APP2}, pages 515--516; \cite{AP2},
Theorem~5.1)]\label{theo.cov}
Let $X\sim\IP\equiv\IPq$ with density $f$
and support $(\alpha,\omega)$. Assume that $X$ has $2k$
finite moments for some fixed $k\in\{1,2,\ldots\}$.
Let $g\dvtx (\alpha,\omega)\to\R$ be any function such that
$g\in C^{k-1}(\alpha,\omega)$, and assume that the function
\[
g^{(k-1)}(x):=\frac{\ud^{k-1}}{\ud{x}^{k-1}}g(x)
\]
is absolutely
continuous in $(\alpha,\omega)$ with a.s. derivative $g^{(k)}$.
If $\E q^k(X)|g^{(k)}(X)|<\infty$ then
$\E|P_k(X)g(X)|<\infty$, where $P_k$ is
the polynomial
defined by (\textup{\ref{Rodrigues}}) of
Theorem~\textup{\ref{theo.polynomials}},
and the following covariance identity holds:
\begin{equation}\label{cov.id}
\E P_k(X) g(X) = \E q^{k}(X)g^{(k)}(X).
\end{equation}
\end{theo}

It should be noted that when we claim that
$h\dvtx (\alpha,\omega)\to\R$ is an absolutely continuous
function with a.s. derivative $h'$ we mean that there exists
a Borel measurable function $h'\dvtx (\alpha,\omega)\to\R$ such that
$h'$ is integrable in every finite subinterval
$[x,y]$ of $(\alpha,\omega)$, and
\[
\int_{x}^y h'(t)\,\ud{t}=h(y)-h(x)\qquad
\mbox{for all compact intervals } [x,y]\subseteq(\alpha,\omega).
\]

%
\begin{cor}[(\cite{APP2}, equation (3.5),
page 516; \cite{AP2}, Corollary~5.1)]
\label{cor.orth}
Let $X\sim\IP\equiv\IPq$.
Assume that for some
$n\in\{1,2,\ldots\}$,
$\E|X|^{2n}<\infty$ or, equivalently,
$\delta<1/(2n-1)$.
Then, the polynomials
defined by (\textup{\ref{Rodrigues}}) of
Theorem~\textup{\ref{theo.polynomials}} satisfy the orthogonality
condition
\begin{eqnarray}\label{orth}
\E\bigl[P_k(X) P_m(X)\bigr] &=&
\delta_{k,m} k! \E q^k(X) \prod
_{j=k-1}^{2k-2} (1-j\delta)\nonumber\\[-8pt]\\[-8pt]
& =&\delta_{k,m} k!
c_k(\delta)\E q^k(X),\qquad  k,m\in\{0,1,\ldots,n\},\nonumber
\end{eqnarray}
where $\delta_{k,m}$ is Kronecker's delta and where an empty product
should be treated as one.
\end{cor}
\setcounter{rem}{0}
\begin{rem}\label{rem.a1}
The orthogonality of $P_k$ and
$P_m$, $k\neq m$, $k,m\in\{0,1,\ldots,n\}$,
remains valid even if
$\delta\in[\frac{1}{2n-1},\frac{1}{2n-2})$; in this case, however,
$P_n\notin L^2(\R,X)$ since $\lead(P_n)>0$ and $\E|X|^{2n}=\infty$.
\end{rem}
%
\begin{rem}\label{rem.a2}
In view of Lemma~\ref{lem.moments},
the assumption $\E|X|^{2n}<\infty$ is
equivalent to the condition $\delta<\frac{1}{2n-1}$.
Therefore, for each $k\in\{1,\ldots,n\}$
and for all $j\in\{k-1,\ldots,2k-2\}$
we have $1-j\delta>0$ because
\[
\{k-1,\ldots,2k-2\}\subseteq\{0,1,\ldots,2n-2\}.
\]
Thus, $c_k(\delta)>0$. Since $\Pr[q(X)>0]=1$, $\deg(q)\leq2$ and
$\E|X|^{2n}<\infty$ we conclude that $0<\E q^k(X)<\infty$ for all
$k\in\{0,1,\ldots,n\}$. It follows that the set
$\{\phi_0,\phi_1,\ldots,\phi_n\}\subset L^2(\R,X)$,
where
\begin{eqnarray}\label{orthonormal}
\phi_k(x)&:=&\frac{P_k(x)}{ (k! c_k(\delta)
\E q^k(X) )^{1/2}}\nonumber\\[-8pt]\\[-8pt]
&\hphantom{:}=&
\frac{((-1)^k/f(x))(\ud^k/\ud x^k)[q^k(x)f(x)]}{ (k!
\E q^k(X)\prod_{j=k-1}^{2k-2}(1-j\delta) )^{1/2}},\qquad  k=0,1,\ldots,n,\nonumber
\end{eqnarray}
is an orthonormal basis of all polynomials with degree at most $n$.
By (\ref{lead}), the leading coefficient of $\phi_k$ is
\begin{equation}
\label{lead.orth} \lead(\phi_k) 
= \biggl(
\frac{\prod_{j=k-1}^{2k-2}(1-j\delta)}{k! \E q^k(X)} \biggr)^{1/2} = \biggl(\frac{c_k(\delta)}{k! \E q^k(X)}
\biggr)^{1/2}>0, \qquad k=0,1,\ldots,n.
\end{equation}
The orthonormal system $\{\phi_k\}_{k=0}^n$ is characterized by the
fact that
$\deg(\phi_k)=k$ and $\lead(\phi_k)>0$ for each $k$.
\end{rem}
%
\begin{rem}\label{rem.complete}
The identity (\ref{cov.id}) enables
a convenient calculation of the Fourier coefficients
of any (smooth enough) function
$g$ with $\Var g(X)<\infty$.
More precisely,
if $X\sim\IP\equiv\IPq$
and $\E|X|^{2n}<\infty$ for
some $n\geq1$
then the Fourier coefficients of $g$,
$\alpha_k=\E\phi_k(X)g(X)$,
are given by $\alpha_0=\E g(X)$ and
\begin{equation}\label{Fourier2}
\alpha_k 
=\frac{\E q^k(X)g^{(k)}(X)}{(k!c_k(\delta)\E q^k(X))^{1/2}},\qquad  k=1,2,
\ldots,n,
\end{equation}
provided that $g$ is smooth enough so that
$\E q^k(X) |g^{(k)}(X)|<\infty$ for $k\in\{1,2,\ldots,n\}$;
cf. \cite{APP2}, Theorem~5.1(a).
Here $c_k(\delta)$ is given by (\ref{lead})
and for any $k\in\{1,\ldots,n\}$
(see \cite{AP2}, Corollary~5.3)
\begin{equation}\label{Eq^k}
\E q^{k}(X) =\frac{\prod_{j=0}^{k-1}(1-2j\delta)}{\prod_{j=0}^{k-1}(1-(2j+1)\delta)} \prod
_{j=0}^{k-1}q \biggl(\frac{\mu+j\beta}{1-2j\delta}
\biggr).
\end{equation}
In the particular case where $X\sim\IP$ and
$\delta\leq0$ (i.e., if $X$
is of Normal, Gamma or Beta-type), it follows that
$\E|X|^n<\infty$ for all $n$. Moreover, there exists
an $\varepsilon>0$ such that $\E e^{tX}<\infty$ for
$|t|<\varepsilon$ (see types 1--3 of Table~2.1 in \cite{AP2}).
Hence, the polynomials $\{\phi_k\}_{k=0}^\infty$,
given by~(\ref{orthonormal}) (with $n=\infty$), form a complete
orthonormal system in $L^2(\R,X)$;
see, for example, \cite{BergCrist,APP2}. Therefore, the Fourier
coefficients are easily obtained for any smooth enough function
$g$ such that $\Var g(X)<\infty$ and
$\E q^k(X) |g^{(k)}(X)|<\infty$
for all $k\geq1$. Indeed, in this case we have
\begin{equation}\label{Fourier}
\alpha_k=\E\phi_k(X)g(X) =
\frac{\E q^k(X)g^{(k)}(X)}{(k!c_k(\delta)\E q^k(X))^{1/2}}, \qquad k=0,1,2,\ldots,
\end{equation}
where 
$\E q^k(X)$ is as in (\ref{Eq^k}).
Thus, by Parseval's identity,
the variance of $g$ equals to
(\cite{APP2}, Theorem~5.1(a))
\begin{equation}\label{variance}
\Var g(X)=\sum_{k=1}^{\infty}
\frac{\E^2 q^k(X)g^{(k)}(X)}{k!c_k(\delta)\E q^k(X)},
\end{equation}
with $\E q^k(X)$ given by (\ref{Eq^k}) and
$c_k(\delta)$ by (\ref{lead}).
\end{rem}
%
\begin{theo}[(\cite{AP2}, Theorem~5.2)]
\label{theo.star}
Let $X$ be a random variable with density
$f\sim\IPq\equiv\IP$,
supported in $(\alpha,\omega)$.
Furthermore, assume that $\E|X|^{2n+1}<\infty$
(i.e., $\delta<\frac{1}{2n}$)
for some $n\in\{0,1,\ldots\}$.
Define the random variable
$X_k$ with density $f_k$ given by
\begin{equation}
\label{star2} f_k(x):=\frac{q^k(x)f(x)}{\E q^k(X)}, \qquad \alpha<x<\omega, k=0,1,
\ldots,n.
\end{equation}
Then, $f_k\sim\IPqk$ with (the same) support
$(\alpha,\omega)$,
\begin{equation}\label{eq.qstar2}
\mu_k=\frac{\mu+k\beta}{1-2k\delta} \quad \mbox{and}\quad
q_k(x)=\frac{q(x)}{1-2k\delta},\qquad  \alpha<x<\omega, k=0,1,\ldots,n.
\end{equation}
\end{theo}
%
\begin{theo}[(\cite{AP2}, Theorem~5.3;
cf. \cite{Beale1}, page 207)]\label{theo.derivatives}
If $X\sim\IP$
with
support $(\alpha,\omega)$ and
$\E|X|^{2n}<\infty$ for some $n\geq1$
(i.e., $\delta<\frac{1}{2n-1}$), then
for any $m\in\{1,2,\ldots,n\}$,
\begin{equation}\label{derivatives.higher}
P_{k+m}^{(m)}(x)=C^{(m)}_k(
\delta) P_{k,m}(x), \qquad \alpha<x<\omega, k=0,1,\ldots,n-m,
\end{equation}
where
\begin{equation}\label{derivatives.higher.polynomials}
C^{(m)}_k(\delta):=
\frac{(k+m)!}{k!}(1-2m\delta)^k \prod_{j=k+m-1}^{k+2m-2}(1-j
\delta).
\end{equation}
Here, $P_k$ are the polynomials given by (\textup{\ref{Rodrigues}})
associated with $f$, and $P_{k,m}$ are the corresponding
Rodrigues polynomials of (\ref{Rodrigues}),
associated with the density $f_m(x)=\frac{q^m(x)f(x)}{\E q^m(X)}$,
$\alpha<x<\omega$,
of the random variable $X_m\sim\operatorname{IP}(\mu_m;q_m)$
defined in Theorem~\textup{\ref{theo.star}}, that is,
%
\begin{eqnarray}\label{polyn.orthogonal.gen.star}
P_{k,m}(x)&:=&\frac{(-1)^k}{f_m(x)}
\frac{\ud^k}{\ud
x^k}\bigl[q_m^k(x)f_m(x)\bigr]\nonumber\\[-8pt]\\[-8pt]
&\hphantom{:}=&\frac{(-1)^k}{(1-2m\delta)^k q^m(x)f(x)} \frac{\ud^k}{\ud x^k}\bigl[q^{k+m}(x)f(x)\bigr],
\qquad
\alpha<x<\omega, k=0,1,\ldots,n-m.\nonumber
\end{eqnarray}
\end{theo}
%
\begin{theo}[(\cite{AP2}, Corollary~5.4)]\label{theor.derivatives}
Let $X\sim\IP\equiv\IPq$
and assume that
$\E|X|^{2n}<\infty$ for some fixed $n\geq1$
(i.e., $\delta<\frac{1}{2n-1}$).
Let $\{\phi_k\}_{k=0}^n$ be the orthonormal
polynomials associated with $X$, with $\lead(\phi_k)>0$;
see (\textup{\ref{orthonormal}}), (\textup{\ref{lead.orth}}).
Fix a number $m\in\{0,1,\ldots,n\}$,
and consider the corresponding orthonormal polynomials
$\{\phi_{k,m}\}_{k=0}^{n-m}$,
with $\lead(\phi_{k,m})>0$,
associated with
$X_m\sim f_m=q^m f/\E q^{m}(X)$.
Then,
\begin{equation}\label{orthonormal.derivatives}
\phi_{k+m}^{(m)}(x)=
\nu_k^{(m)} \phi_{k,m}(x), 
\qquad k=0,1,
\ldots,n-m,
\end{equation}
where the constants
$\nu_k^{(m)}=\nu_k^{(m)}(\mu;q)>0$
are given by
\begin{equation}\label{orthonormal.constants}
\nu_k^{(m)}=
\nu^{(m)}_k(\mu;q): = \biggl\{ \frac{((k+m)!/k!)\prod_{j=k+m-1}^{k+2m-2}(1-j\delta)}{\E q^m(X)} \biggr
\}^{1/2},
\end{equation}
with $\E q^m(X)$ as in (\textup{\ref{Eq^k}}) with $m$ in place of $k$.
In particular, setting $\sigma^2=\Var X=\E q(X)$ we have
\begin{eqnarray}\label{eq.paragwgos}
\phi_{k+1}'(x)&=&\frac{\sqrt{(k+1)(1-k\delta)}}{\sigma}
\phi_{k,1}(x) \nonumber\\[-8pt]\\[-8pt]
&=&\sqrt{\frac{(k+1)(1-\delta)(1-k\delta)}{q(\mu)}}\phi_{k,1}(x), \qquad k=0,1,
\ldots,n-1.\nonumber
\end{eqnarray}
%
\end{theo}
\end{appendix}

\section*{Acknowledgement}
Work partially supported
 by the University of Athens Research Grant 70/4/5637.



\printhistory

\end{document}